\documentclass[a4rpaper, 12pt, conference]{ieeeconf}  

\IEEEoverridecommandlockouts                              
\usepackage{mathptmx} 

\usepackage{graphicx} 
\usepackage{times} 
\usepackage{amsmath} 
\usepackage{amssymb}  
\usepackage{datetime}

\newtheorem{thm}{Theorem}

\def\<{\leqslant}           
\def\>{\geqslant}           

\def\wh{\widehat}

\def\mod{\mathrm{mod\,}}   
\def\Res{\mathop{\mathrm{Res\,}}} 
\def\Re{\mathrm{Re\,}}   
\def\Im{\mathrm{Im\,}}   

\def\mR{{\mathbb R}}    
\def\mC{\mathbb{C}}    

\def\Tr{{\rm Tr}}       
\def\rT{{\rm T}}        

\def\bE{\mathbf{E}}    

\def\[[[{[\![\![}
\def\]]]{]\!]\!]}

\def\re{{\rm e}}        
\def\rd{{\rm d}}        

\def\br{\mathbf{r}}
\def\x{\times}
\def\ox{\otimes}

\def\cA{\mathcal{A}}
\def\cB{\mathcal{B}}
\def\cE{\mathcal{E}}

\def\mU{\mathbb{U}}

\title{\bf
Physical Realizability and Mean Square Performance of Translation Invariant Networks of Interacting Linear Quantum Stochastic Systems}

\author{Igor G. Vladimirov, \quad Ian R. Petersen\thanks{This work is supported by the Australian Research Council.
The authors are with the School of Engineering and Information Technology,
University of New South Wales Canberra, ACT 2600, Australia. E-mail: {\tt igor.g.vladimirov@gmail.com, i.r.petersen@gmail.com}.}
}

\onecolumn
\begin{document}

\maketitle

\thispagestyle{empty}

\begin{abstract}
This paper is concerned with translation invariant networks of linear quantum stochastic systems with nearest neighbour interaction mediated by boson fields.  The systems are associated with sites of a one-dimensional chain or a multidimensional lattice and are governed by coupled linear quantum stochastic differential equations (QSDEs). Such interconnections of open quantum  systems are relevant, for example, to the phonon theory of crystalline solids, atom trapping in optical lattices and quantum metamaterials.   In order to represent a large-scale open quantum harmonic oscillator, the coefficients of the coupled QSDEs must satisfy certain physical realizability conditions. These are established in the form of matrix algebraic equations for the parameters of an individual building block of the network and its interaction with the neighbours and external fields.
 We also discuss the computation of mean square performance functionals with block Toeplitz weighting matrices   for such systems in the thermodynamic limit per site for unboundedly increasing fragments of the lattice.
\end{abstract}

\section{\bf Introduction}

Large-scale networks of interacting dynamical systems with translational symmetry are ubiquitous in the physical world. A natural example is provided by crystalline solids, where the spatial periodicity in the arrangement of atoms results from their interaction and plays an important role in the thermodynamic and mechanical properties of solids  (such as heat capacity and speed of sound)  as reflected in the phonon theory  \cite{S_1990}.

Translation  invariant interconnections  are exploited in modern engineering applications such as
communication networks, vehicle  platooning and metamaterials   \cite{VBSH_2006}. The latter are exemplified by split ring resonator arrays which are artificially fabricated media with a negative refraction index. The unusual electrodynamic properties of such materials are achieved, not only due to the common structure of individual building blocks, but also the translation invariance of their interconnection. Moreover, the last feature is crucial for the nontrivial collective response of the whole system to time-varying  electromagnetic fields.

These ideas are taken from the classical macroscopic scale to the quantum level in quantum metamaterials \cite{QSMGH_2011} which are organised as one, two or three-dimensional \cite{Z_2012} periodic arrays of coherently coupled quantum systems. The latter form a fully quantum composite system which does not involve measurements.
Active research into this kind of artificial materials is inspired by qualitatively new  properties  of light-matter interaction which 
unveil previously hidden resources.
An  example is provided by artificial crystals of atoms trapped at nodes of an optical lattice which can be controlled by external fields and used for entanglement generation \cite{CBFFFRCIP_2013} or as a quantum memory \cite{NDMRLLWJ_2010}.

In light of the emerging technology of quantum metamaterials,  the present paper is concerned with the modelling and analysis of the dynamics of translation invariant networks of interacting linear quantum stochastic
systems which represent open quantum harmonic oscillators. In particular, we are concerned with physical realizability (PR) and mean square performance of this class of large-scale quantum systems from the viewpoint of quantum linear systems theory \cite{P_2010}.

In the present setting, the quantum systems form a one-dimensional chain or are associated with sites of a multidimensional lattice and are governed by a set of coupled linear quantum stochastic differential equations (QSDEs) in the framework of the Hudson-Parthasarathy noncommutative version of the Ito calculus \cite{HP_1984,P_1992}. Although these QSDEs look similar to the classical Ito SDEs widely used in linear stochastic control  theory \cite{AM_1989,KS_1972}, their coefficients 
must satisfy certain PR conditions \cite{JNP_2008,NJP_2009,SP_2012}.

The PR constraints are closely related with the postulate of a unitary evolution \cite{M_1998,S_1994} for isolated quantum systems (for example, those formed from a system of interest and its environment which may involve other quantum systems and external fields). Such an evolution preserves canonical commutation relations (CCRs) between quantum variables and is specified by the ``energetics'' of the underlying system and its interaction with the surroundings.

The dynamics of interconnected quantum stochastic systems are governed by QSDEs whose drift and diffusion terms are expressed in terms of scattering, coupling and Hamiltonian operators \cite{GJ_2009}.  For linear quantum stochastic systems,   the PR conditions reflect dynamic equivalence to an open quantum harmonic oscillator \cite{EB_2005,GZ_2004} whose dynamic variables satisfy CCRs. Its  Hamiltonian is quadratic and the coupling operator  is linear with respect to the dynamic variables. The resulting PR conditions are organised as quadratic constraints on the state-space matrices of the QSDEs \cite{JNP_2008,NJP_2009,SP_2012}.

Using the previous results on PR of linear QSDEs, we take advantage of the specific structure of coupled QSDEs for translation  invariant quantum networks with nearest neighbour interaction. This allows PR constraints to be established here in the form of matrix algebraic equations for the parameters of an individual building block of the network and its coupling to its neighbours and external fields.

Adapting the performance criteria used in the Coherent Quantum Linear Quadratic Gaussian (CQLQG) control/filtering problems \cite{NJP_2009,VP_2013a,VP_2013b}, we also discuss mean square functionals with block Toeplitz weighting matrices whose structure reflects the translation invariance of the quantum network. Under a stability condition, we compute the steady-state value of such a  functional per site for unboundedly increasing fragments of the network. This corresponds to the thermodynamic limit in equilibrium statistical mechanics \cite{R_1978}. The results of the paper can be used for the development of decentralised CQLQG controllers for large-scale quantum networks.
%
%

The paper is organized as follows. Section \ref{sec:chain} describes a one-dimensional chain of interacting quantum systems. Section~\ref{sec:ztrans} introduces spatial Fourier transforms of 
quantum processes. These are employed  in Section~\ref{sec:PR} to establish PR conditions for the governing QSDEs. 
Section \ref{sec:quadro} computes mean square performance functionals with block Toeplitz weights in the thermodynamic limit. 
Section \ref{sec:2dlattice} outlines an extension of the results to the multivariate case.

\section{\bf One-dimensional chain of linear quantum stochastic systems}\label{sec:chain}

Consider a one-dimensional chain of identical open quantum systems with nearest neighbour interaction. The interaction is  arranged in a coherent (measurement-free) fashion and is  mediated by quantum fields which propagate through quantum channels shown as directed edges in  Fig.~\ref{fig:chain}.
\begin{figure}[htbp]
\centering
\unitlength=1mm
\linethickness{0.4pt}
\begin{picture}(105.00,85.00)
    \multiput(0, 40)(40, 0){3}{
        \framebox(20,20)[cc]{{$F$}}
        \put(-10,35){\vector(0,-1){15}}
        \put(-10,0){\vector(0,-1){15}}

        \put(-20,5){\vector(-1,0){20}}
        \put(20,5){\vector(-1,0){20}}
        \put(-40,15){\vector(1,0){20}}
        \put(0,15){\vector(1,0){20}}
    }
    \put(13,80){\makebox(0,0)[cc]{{$w_{k-1}$}}}
    \put(53,80){\makebox(0,0)[cc]{{$w_{k}$}}}
    \put(93,80){\makebox(0,0)[cc]{{$w_{k+1}$}}}

    \put(13,20){\makebox(0,0)[cc]{{$r_{k-1}$}}}
    \put(53,20){\makebox(0,0)[cc]{{$r_{k}$}}}
    \put(93,20){\makebox(0,0)[cc]{{$r_{k+1}$}}}

    \put(-9,60){\makebox(0,0)[cc]{{$y_{k-2}^+ = u_{k-1}^+$}}}
    \put(-9,40){\makebox(0,0)[cc]{{$u_{k-2}^-=y_{k-1}^-$}}}

    \put(-23,52){\makebox(0,0)[cc]{$\vdots$}}
    \put(127,52){\makebox(0,0)[cc]{$\vdots$}}

    \put(33,60){\makebox(0,0)[cc]{{$y_{k-1}^+=u_k^+$}}}
    \put(33,40){\makebox(0,0)[cc]{{$u_{k-1}^-=y_k^- $}}}

    \put(73,60){\makebox(0,0)[cc]{{$y_k^+=u_{k+1}^+$}}}
    \put(73,40){\makebox(0,0)[cc]{{$u_k^-=y_{k+1}^- $}}}

    \put(114,60){\makebox(0,0)[cc]{{$y_{k+1}^+=u_{k+2}^+$}}}
    \put(114,40){\makebox(0,0)[cc]{{$u_{k+1}^-=y_{k+2}^- $}}}

\end{picture}\vskip-1cm
\caption{A fragment of a one-dimensional chain of linear quantum stochastic systems with nearest neighbour interaction. The subsystems and fields are numbered from left to right. The ``$+$'' and ``$-$'' superscripts indicate field  propagation in the  positive and negative directions, respectively.}
\label{fig:chain}
\end{figure}
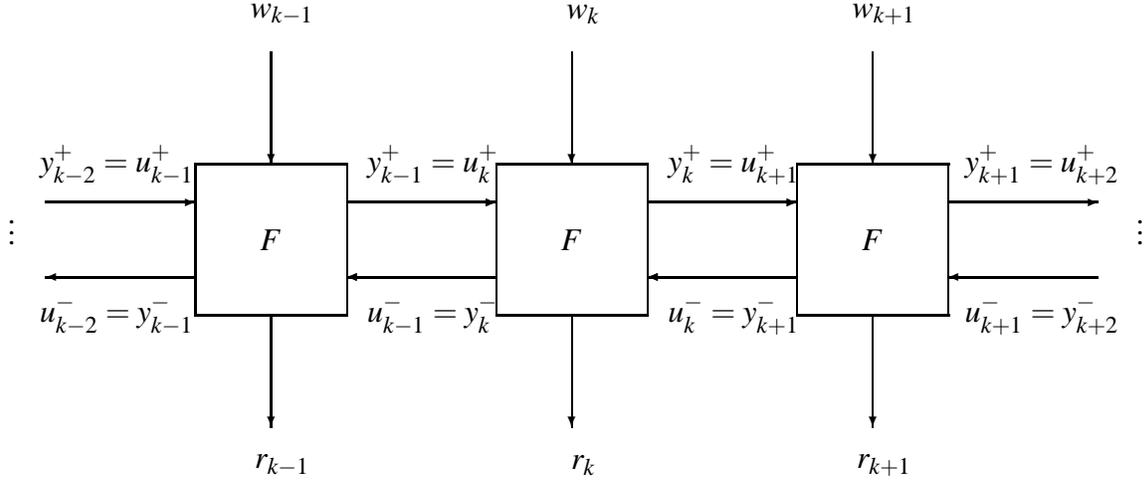
This infinite interconnection is a translation invariant quantum feedback network which, in practice, can be implemented as its finite fragment.
Suppose the fragment of the chain consists of $N$ building blocks which are numbered by integers $k = 0, \ldots, N-1$.
The $k$th block is a linear quantum stochastic system with ``rightward'' input and  output $u_k^+$ and $y_k^+$ of common dimension $m_+$ and ``leftward''  input and  output $u_k^-$ and $y_k^-$ of common dimension $m_-$; see Fig.~\ref{fig:block}.
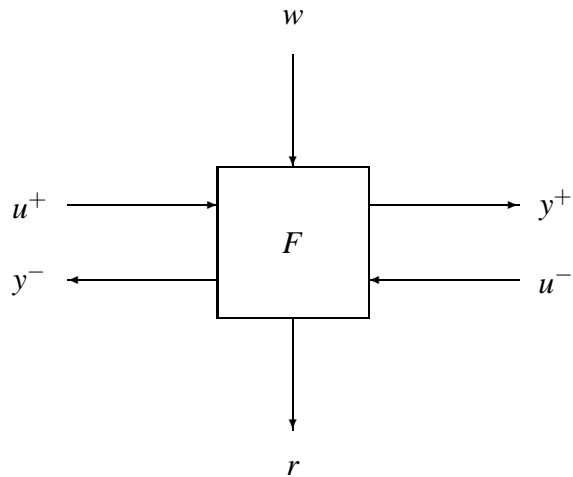
\begin{figure}[htbp]
\centering
\unitlength=1mm
\linethickness{0.4pt}
\begin{picture}(70.00,60.00)
    \put(20, 17){
        \framebox(20,20)[cc]{{$F$}}
        \put(-10,35){\vector(0,-1){15}}
        \put(-10,0){\vector(0,-1){15}}

        \put(-20,5){\vector(-1,0){20}}
        \put(20,5){\vector(-1,0){20}}
        \put(-40,15){\vector(1,0){20}}
        \put(0,15){\vector(1,0){20}}
        \put(-10,40){\makebox(0,0)[cc]{{$w$}}}

        \put(-10,-20){\makebox(0,0)[cc]{{$r$}}}

        \put(-45,15){\makebox(0,0)[cc]{{$u^+$}}}
        \put(-45,5){\makebox(0,0)[cc]{{$y^-$}}}

        \put(25,15){\makebox(0,0)[cc]{{$y^+$}}}
        \put(25,5){\makebox(0,0)[cc]{{$u^-$}}}
    }
\end{picture}\vskip5mm
\caption{A common building block of the chain.}
\label{fig:block}
\end{figure}
The system  is endowed with an $n$-dimensional vector $x_k$ of dynamic variables and is coupled to external input fields  which are modelled by an $m_0$-dimensional  quantum Wiener process $w_k$ on a boson Fock space \cite{P_1992}. As shown in Figs.~\ref{fig:chain}, \ref{fig:block},  the systems can also have external output fields $r_k$. However, they will be taken into consideration elsewhere. The entries of the vectors $x_k$, $w_k$, $u_k^{\pm}$, $y_k^{\pm}$  are self-adjoint operators on appropriate complex separable Hilbert spaces evolving in time in accordance with the Heisenberg picture of quantum dynamics \cite{M_1998}. Unless specified otherwise, vectors are organised as columns. The joint quantum Ito table of the quantum Wiener processes $w_k$ is given by
\begin{equation}
\label{ww}
    \rd w_j \rd w_k^{\rT} = \delta_{jk}\Omega\rd t,
    \qquad
    \Omega := I_{m_0} + i J.
\end{equation}
Here, the transpose $(\cdot)^{\rT}$ acts on vectors of operators as if the latter were scalars. Also,  $i:= \sqrt{-1}$ is the imaginary unit,
$\delta_{jk}$ is the Kronecker delta, $I_{m_0}$ is the identity matrix of order $m_0$, and
$J$ is a real antisymmetric matrix of order $m_0$. The matrix $J$ has spectral radius $\br(J)\< 1$ (thus ensuring the positive semi-definiteness  of the quantum Ito matrix $\Omega \succcurlyeq 0$) and
specifies the cross commutations between the boson fields as
\begin{equation}
\label{wCCR}
    [\rd w_j, \rd w_k^{\rT}]
    =
    2i
    \delta_{jk}
    J\rd t,
\end{equation}
where $[\alpha, \beta^{\rT}]:= \alpha\beta^{\rT} - (\beta\alpha^{\rT})^{\rT}$ is the commutator matrix.
Usually, the quantum noise dimension $m_0$ is even, and
$    J = I_{m_0/2} \ox
        {\begin{bmatrix}
        0 & 1\\
        -1  & 0
    \end{bmatrix}}
$,
where $\ox$ is the Kronecker product of matrices. In view of (\ref{ww}), (\ref{wCCR}), the quantum Wiener processes $w_k$ for different component systems are uncorrelated and commuting.
Now, the $k$th system in the chain is governed by the following linear QSDEs with constant coefficients
\begin{align}
\label{dx}
    \rd x_k
    &=
    A x_k \rd t + B \rd w_k + E \rd u_k,\\
\label{dy}
    \rd y_k
    &=
    C x_k \rd t + D \rd w_k.
\end{align}
Here,
\begin{equation}
\label{uy}
    u_k
    :=
    {\begin{bmatrix}
        u_k^+\\
        u_k^-
    \end{bmatrix}},
    \qquad
    y_k :=
    {\begin{bmatrix}
        y_k^+\\
        y_k^-
    \end{bmatrix}}
\end{equation}
are the input and output quantum processes, each having dimension
$
    m:= m_++m_-
$,
and
\begin{equation}
\label{CDE}
    C :=
    {\begin{bmatrix}
        C_+ \\
        C_-
    \end{bmatrix}},
    \qquad
    D
    :=
    {\begin{bmatrix}
        D_{+} \\
        D_{-}
    \end{bmatrix}},
    \qquad
    E :=
    {\begin{bmatrix}
        E_+ & E_-
    \end{bmatrix}},
\end{equation}
where $A \in \mR^{n\x n}$, $B\in \mR^{n\x m_0}$, $C_{\pm} \in \mR^{m_{\pm} \x n}$, $D_{\pm} \in \mR^{m_{\pm} \x m_0}$, $E_{\pm}\in \mR^{n\x m_{\pm}}$  are appropriately dimensioned real matrices. In view of the interconnection of the systems in the chain (see Fig.~\ref{fig:chain}),  the QSDEs (\ref{dx}), (\ref{dy}) are complemented by the algebraic relations for the rightward and leftward inputs and outputs of the adjacent systems:
\begin{equation}
\label{connect}
    y_{k-1}^+ = u_k^+,
    \qquad
    u_k^- = y_{k+1}^-.
\end{equation}
Furthermore, the  set of equations (\ref{dx}), (\ref{dy}), (\ref{connect}) for the fragment of the chain must be equipped with boundary conditions for the $0$th and $(N-1)$th systems. As a variant of such conditions, it can be assumed that the boundary inputs $u_0^+$ and $u_{N-1}^-$ are additional uncorrelated and commuting quantum Wiener processes.
However, in order to simplify the analysis at this stage, we will use the periodic boundary conditions (PBCs)
\begin{equation}
\label{PBC}
    u_0^+ = y_{N-1}^+,
    \qquad
    u_{N-1}^- = y_0^-,
\end{equation}
as shown in Fig.~\ref{fig:PBC}.
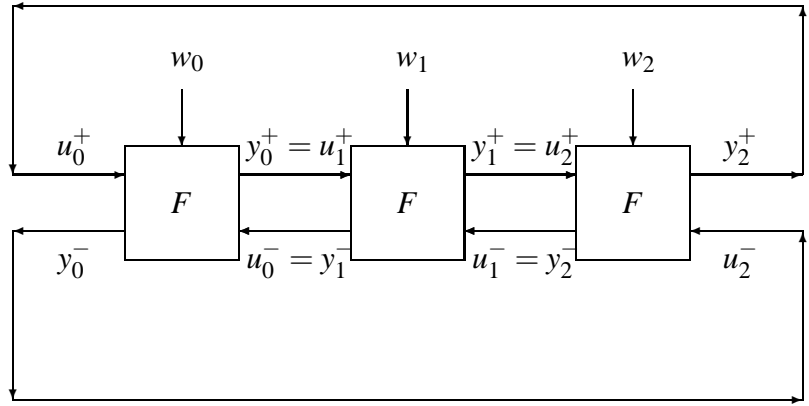
\begin{figure}[htbp]
\centering
\unitlength=0.75mm
\linethickness{0.4pt}
\begin{picture}(105.00,95.00)
    \multiput(0, 40)(40, 0){3}{
        \framebox(20,20)[cc]{{$F$}}
        \put(-10,30){\vector(0,-1){10}}

        \put(-20,5){\vector(-1,0){20}}
        \put(20,5){\vector(-1,0){20}}
        \put(-40,15){\vector(1,0){20}}
        \put(0,15){\vector(1,0){20}}
    }
    \put(13,75){\makebox(0,0)[cc]{{$w_0$}}}
    \put(53,75){\makebox(0,0)[cc]{{$w_1$}}}
    \put(93,75){\makebox(0,0)[cc]{{$w_2$}}}


    \put(-7,60){\makebox(0,0)[cc]{{$u_0^+$}}}
    \put(-7,40){\makebox(0,0)[cc]{{$y_0^-$}}}

    \put(33,60){\makebox(0,0)[cc]{{$y_0^+=u_1^+$}}}
    \put(33,40){\makebox(0,0)[cc]{{$u_0^-=y_1^- $}}}

    \put(73,60){\makebox(0,0)[cc]{{$y_1^+=u_2^+$}}}
    \put(73,40){\makebox(0,0)[cc]{{$u_1^-=y_2^- $}}}

    \put(111,60){\makebox(0,0)[cc]{{$y_2^+$}}}
    \put(111,40){\makebox(0,0)[cc]{{$u_2^-$}}}

    \put(122,55){\vector(0,1){30}}
    \put(122,85){\vector(-1,0){140}}
    \put(-18,85){\vector(0,-1){30}}

    \put(-18,45){\vector(0,-1){30}}
    \put(-18,15){\vector(1,0){140}}
    \put(122,15){\vector(0,1){30}}
\end{picture}\vskip-5mm
\caption{An illustration of  the PBCs for a fragment of the chain with $N = 3$.}
\label{fig:PBC}
\end{figure}
In this case, the fragment of the chain acquires a ring topology, 
and the interconnection rules (\ref{connect}) take a more unified form
\begin{equation}
\label{connectPBC}
    u_k
    =
    {\begin{bmatrix}
        I_{m_+} & 0\\
        0 & 0
    \end{bmatrix}}
    y_{k-1}
    +
    {\begin{bmatrix}
        0 & 0\\
        0 & I_{m_-}
    \end{bmatrix}}
    y_{k+1},
    \qquad
    k=0, \ldots, N-1,
\end{equation}
where $k\pm 1$  are calculated modulo $N$. This representation is convenient for the harmonic analysis of the quantum network in the spatial frequency domain.

\section{\bf Spatial Fourier transforms}\label{sec:ztrans}

Consider the $z$-transform of the quantum processes $x_k$, $w_k$, $u_k$, $y_k$ in (\ref{dx})--(\ref{uy}) over the ``spatial'' subscript $k = 0, \ldots, N-1$ which numbers the systems in the fragment of the chain:
\begin{align}
\label{XW}
    X_z(t)
    & :=
    \sum_{k=0}^{N-1} z^{-k}x_k(t),
    \quad
    W_z(t)
    :=
    \sum_{k=0}^{N-1} z^{-k}w_k(t),\\
\label{U}
    U_z(t)
    & :=
    {\begin{bmatrix}
        U_z^+(t)\\
        U_z^-(t)
    \end{bmatrix}},
    \qquad
    U_z^{\pm}(t)
    :=
    \sum_{k=0}^{N-1} z^{-k} u_k^\pm(t),\\
\label{Y}
    Y_z(t)
    & :=
    {\begin{bmatrix}
        Y_z^+(t)\\
        Y_z^-(t)
    \end{bmatrix}},
    \qquad\
    Y_z^{\pm}(t)
    :=
    \sum_{k=0}^{N-1}
    z^{-k}
    y_k^{\pm}(t),
\end{align}
where $z$ is a nonzero complex parameter.  Note that, being linear combinations of self-adjoint operators with complex coefficients, the entries of the vectors $X_z(t)$, $W_z(t)$, $U_z(t)$, $Y_z(t)$  are not self-adjoint.
By applying $z$-transforms to the linear QSDEs (\ref{dx}), (\ref{dy}) and using (\ref{XW})--(\ref{Y}), it follows that the quantum processes $X_z$, $W_z$, $U_z$, $Y_z$, as functions of time $t$, satisfy QSDEs with the same coefficients:
\begin{align}
\label{dX}
    \rd X_z
    &=
    A X_z \rd t + B \rd W_z + E \rd U_z,\\
\label{dY}
    \rd Y_z
    &=
    C X_z \rd t + D \rd W_z,
\end{align}
where the time arguments are omitted for brevity.
Furthermore, in the framework of the PBCs (\ref{PBC}), the equalities (\ref{connectPBC}) imply that the $z$-transforms $U_z$, $Y_z$ in (\ref{U}), (\ref{Y}) are related by
\begin{align}
\nonumber
    U_z
    =&
    {\begin{bmatrix}
        I_{m_+} & 0\\
        0 & 0
    \end{bmatrix}}
    \left(
        z^{-1}Y_z + \left(1 - z^{-N}\right)y_{N-1}
    \right)\\
\nonumber
    & +
    {\begin{bmatrix}
        0 & 0\\
        0 & I_{m_-}
    \end{bmatrix}}
    \left(
        z Y_z + z\left(z^{-N}-1\right)y_0
    \right)\\
\label{UY}
    = &
    {\begin{bmatrix}
        z^{-1}I_{m_+} & 0\\
        0 & zI_{m_-}
    \end{bmatrix}}
    Y_z
    +
    \left(1 - z^{-N}\right)
    {\begin{bmatrix}
        y_{N-1}^+\\
        -zy_0^-
    \end{bmatrix}}.
\end{align}
In particular, the boundary outputs $y_0^-$ and $y_{N-1}^+$ of the chain fragment make no contribution  to (\ref{UY}) when $z$ belongs to the set of $N$th roots of unity
\begin{equation}
\label{roots}
    \mU_N:=
    \big\{
        \re^{2\pi i\ell/N}:\
        \ell = 0, \ldots, N-1
    \big\}.
\end{equation}
In this case, the quantum processes $X_z$, $W_z$, $U_z$, $Y_z$ in (\ref{XW})--(\ref{Y}), as functions of $\ell$ in (\ref{roots}),  become the spatial discrete Fourier transforms (DFT) of the quantum processes $x_k$, $w_k$,  $u_k$,  $y_k$  over $k = 0, \ldots, N-1$, with $2\pi\ell/N$ playing the role of a wavenumber.
Recall that the set $\mU_N$ is a multiplicative group which is isomorphic to the additive group of residues modulo $N$. Since $\mU_N$ is a subset of the  unit circle
$$
    \mU:= \{z \in \mC:\ |z| = 1\}
$$
in the complex plane, the inversion in $\mU_N$ is equivalent to the complex conjugation: $z^{-1} = \overline{z}$.
In view of (\ref{UY}), for any $z\in \mU_N$, the quantum processes $U_z$, $Y_z$  are related by a static (time-independent) unitary transformation as
\begin{equation}
\label{K}
    U_z
    =
    {\begin{bmatrix}
        z^{-1}Y_z^+\\
        zY_z^-
    \end{bmatrix}}
    =
    K_z
    Y_z,
    \qquad
    K_z
    :=
    {\begin{bmatrix}
        z^{-1}I_{m_+} & 0\\
        0 & zI_{m_-}
    \end{bmatrix}},
\end{equation}
where use is also made of (\ref{Y}). The relationship (\ref{K}) describes two feedback loops with phase shift factors $z^{\pm1}$ (see Fig.~\ref{fig:blockz})
\begin{figure}[htbp]
\centering
\unitlength=0.75mm
\linethickness{0.4pt}
\begin{picture}(60.00,80.00)
    \put(20, 25){
        \framebox(20,20)[cc]{{$F$}}
        \put(-15,40){\framebox(10,10)[cc]{{$z^{-1}$}}}
        \put(-15,-30){\framebox(10,10)[cc]{{$z$}}}
        \put(-10,28){\vector(0,-1){8}}

        \put(-20,5){\vector(-1,0){20}}
        \put(20,5){\vector(-1,0){20}}

        \put(-40,15){\vector(1,0){20}}
        \put(0,15){\vector(1,0){20}}
        \put(20,15){\vector(0,1){30}}
        \put(20,45){\vector(-1,0){25}}
        \put(-15,45){\vector(-1,0){25}}
        \put(-40,45){\vector(0,-1){30}}

        \put(-40,5){\vector(0,-1){30}}
        \put(-40,-25){\vector(1,0){25}}
        \put(-5,-25){\vector(1,0){25}}
        \put(20,-25){\vector(0,1){30}}

        \put(-10,32){\makebox(0,0)[cc]{{$W_z$}}}


        \put(-45,15){\makebox(0,0)[rc]{{$U_z^+$}}}
        \put(-45,5){\makebox(0,0)[rc]{{$Y_z^-$}}}

        \put(25,15){\makebox(0,0)[lc]{{$Y_z^+$}}}
        \put(25,5){\makebox(0,0)[lc]{{$U_z^-$}}}
    }
\end{picture}\vskip1cm
\caption{A block diagram of the QSDEs (\ref{dX}), (\ref{dY}) combined with (\ref{K}).}
\label{fig:blockz}
\end{figure}
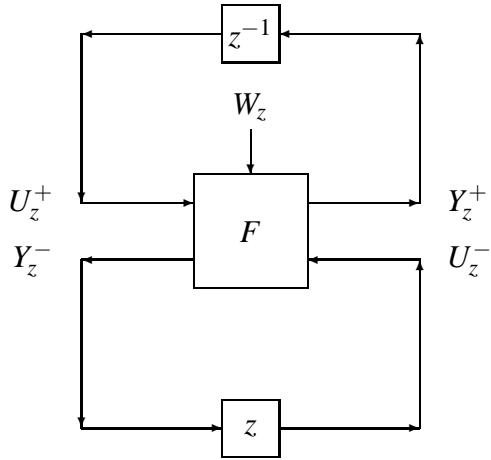
and  allows $U_z$ to be eliminated from the QSDE (\ref{dX}) by substituting $\rd Y_z$ from (\ref{dY}) into $\rd U_z = K_z \rd Y_z$, which  yields  the QSDE
\begin{equation}
\label{dX1}
    \rd X_z
    =
    \cA_z X_z \rd t + \cB_z \rd W_z,
\end{equation}
with
\begin{equation}
\label{cAB}
    \cA_z := A + EK_z  C,
    \qquad
    \cB_z := B + E K_z D.
\end{equation}
This procedure can be  justified by the elimination of edges in a quantum feedback network in the zero time delay limit \cite{GJ_2009}.
Since the diagonal matrix $K_z$ in (\ref{K}) satisfies
$
    K_{1/z} = \overline{K_z}
$ for any $z$ on the unit circle, and the state-space matrices of the QSDEs (\ref{dx}), (\ref{dy}) are real, then (\ref{cAB}) implies that
\begin{equation}
\label{star}
    \cA_z^* = \cA_{1/z}^{\rT},
    \qquad
    \cB_z^* = \cB_{1/z}^{\rT},
    \qquad
    z \in \mU,
\end{equation}
where $(\cdot)^*:= (\overline{(\cdot)})^{\rT}$ is the complex conjugate transpose of a matrix.
Note that $W_z(t)$ in (\ref{XW}), as functions of time $t$,   are quantum Wiener processes whose  joint quantum Ito table, in view of (\ref{ww}),  is computed as
\begin{equation}
\label{WW}
    \rd W_z \rd W_v^{\dagger}
    =
    \sum_{j,k=0}^{N-1}
    z^{-j}v^k
    \rd w_j \rd w_k^{\rT}
    =
    \sum_{k=0}^{N-1}
    \Big(\frac{v}{z}\Big)^k
        \Omega \rd t
    =
    N \delta_{zv}\Omega \rd t
\end{equation}
for any roots of unity $z, v \in \mU_N$ from (\ref{roots}), with $(\cdot)^{\dagger}:= ((\cdot)^{\#})^{\rT}$ the transpose
of the entry-wise adjoint $(\cdot)^{\#}$ of a matrix of operators.
  Here, use is also made of the bilinearity of the commutator, the group property of $\mU_N$, and the identity
$
    \sum_{k=0}^{N-1}
    \zeta^k
    =
    N \delta_{\zeta 1}
%
$ for all $\zeta \in \mU_N$.
Note that the right-hand side of  (\ref{WW}) is nonzero only for $z=v$, and hence, $W_z$ and $W_v$ are uncorrelated and commuting for $z\ne v$.
In accordance with (\ref{wCCR}), the relation (\ref{WW}) implies
\begin{align}
\label{WCCR}
    [\rd W_z, \rd W_v^{\dagger}]
    =
    2i \delta_{zv} N  J \rd t.
\end{align}
Therefore, instead of the original system of coupled QSDEs  (\ref{dx}), (\ref{dy}), (\ref{connect})
(whose number increases with the chain fragment size $N$), we have obtained the algebraically closed QSDEs (\ref{dY}),  (\ref{dX1}) parameterized by $z \in \mU$. For any fixed  $z$, these QSDEs describe a  linear quantum stochastic system $F_z$ with the quadruple of state-space matrices $\cA_z$, $\cB_z$, $C$, $D$ in (\ref{cAB}) and input and output fields $W_z$ and $Y_z$.
%
%
This family of quantum systems (in which every member can be regarded independently of the others) encodes the dynamics of the whole network. In fact, the systems $F_z$, considered for different values of $z$, are analogous to the independent spatial modes of vibration in  the phonon theory of crystal lattices \cite{S_1990}.

\section{\bf Physical realizability conditions}\label{sec:PR}

The dynamic variables of the quantum systems in the chain are self-adjoint operators which act initially on  copies of a common Hilbert space.  They are assumed to satisfy the CCRs of a quantum harmonic oscillator \cite{M_1998} for each of the constituent systems and commute for different systems:
\begin{equation}
\label{xCCR}
    [x_j, x_k^{\rT}] = 2i\delta_{jk}\Theta,
    \qquad
    0\<
    j,k< N.
\end{equation}
Here,  $\Theta$ is a real antisymmetric matrix of order $n$.
A necessary condition for the set of linear QSDEs (\ref{dx}), (\ref{dy}) to be PR \cite{JNP_2008,NJP_2009,SP_2012}  as an open quantum harmonic  oscillator is the preservation in time of the CCRs (\ref{xCCR}) for the state variables along with the commutativity of the state and output
\begin{equation}
\label{xyCCR}
    [x_k, y_k^{\rT}] = 0,
    \qquad
    k = 0,\ldots, N-1.
\end{equation}
The latter property reflects the non-demolition nature of coherent quantum feedback interconnections whereby the output fields of open quantum systems  behave like ideal measurements with respect to the internal dynamic variables of the systems \cite{B_1989}. Note that, at the initial moment of time,  the state and output variables (as operators on different initial Hilbert spaces) commute for any pair of the constituent systems:
\begin{equation}
\label{xyCCR0}
    [x_j(0), y_k(0)^{\rT}] = 0,
    \qquad
    0\< j,k <N.
\end{equation}
Similarly to (\ref{WCCR}),
the CCRs (\ref{xCCR}) can be equivalently represented in terms of the DFT $X_z$ in (\ref{XW}) as
\begin{align}
\nonumber
    [X_z, X_v^{\dagger}]
    &= \sum_{j,k=0}^{N-1} z^{-j} v^{k}  [x_j, x_k^{\rT}]\\
\label{XCCR}
    & =
    2i\Theta \sum_{k=0}^{N-1} (v/z)^k
    =
    2i\delta_{zv}N\Theta,
    \qquad
    z,v \in \mU_N.
\end{align}
The equivalence between (\ref{xCCR}) and (\ref{XCCR}) follows from the fact that the commutator matrices in \begin{equation}
\label{XXxx}
([X_z,X_v^{\dagger}])_{z, v \in \mU_N} = N \Phi_N\, ([x_j,x_k^{\rT}])_{0\< j, k< N}\, \Phi_N^*
\end{equation}
are related by the  unitary matrix from the DFT $    \Phi_N
    :=
    \frac{1}{\sqrt{N}}
    \big(
        \re^{-2\pi i \ell \mu/N}
    \big)_{0\< \ell, \mu< N}
$.
Application
of the inverse DFT to (\ref{XW}) and (\ref{Y}) with
\begin{equation}
\label{invFourier}
    x_k= \frac{1}{N} \sum_{z\in \mU_N} z^k X_z,
    \qquad
    y_k= \frac{1}{N} \sum_{v\in \mU_N} v^k Y_v
\end{equation}
yields
$$
    [x_k, y_k^{\rT}]
    =
    \frac{1}{N^2}
    \sum_{z,v\in \mU_N}
    (z/v)^k [X_z, Y_v^{\dagger}]
    =
    \frac{1}{N^2}
    \sum_{r\in \mU_N}
    r^k
    \sum_{v\in \mU_N}
    [X_{rv}, Y_v^{\dagger}],
$$
and hence,
\begin{equation}
\label{xyXY}
    \frac{1}{N}
    \sum_{v\in \mU_N}
    [X_{rv}, Y_v^{\dagger}]
    =
    \sum_{k=0}^{N-1}
    r^{-k}    [x_k, y_k^{\rT}].
\end{equation}
Since the left-hand side of (\ref{xyXY}) is the DFT of the sequence of commutators $[x_k,y_k^{\rT}]$, the fulfillment of (\ref{xyCCR}) is equivalent to
\begin{equation}
\label{xyXY0}
    \sum_{v\in \mU_N}
    [X_{rv}, Y_v^{\dagger}]
    =
    0,
    \qquad
    r\in \mU_N.
\end{equation}
Note that (\ref{xyXY0}) is satisfied at the initial moment of time due to (\ref{xyCCR0}) whereby
\begin{equation}
\label{XYCCR0}
  [X_z(0),Y_v(0)^{\dagger}] = 0,
  \qquad
  z,v \in \mU_N.
\end{equation}

\begin{thm}
\label{th:CCR}
The CCRs (\ref{xCCR}), (\ref{xyCCR}) for the dynamic and output variables of the systems $0, \ldots, N-1$  in the chain are preserved in time if and only if the state-space matrices in (\ref{cAB}) satisfy
\begin{align}
\label{Thetadot0}
    \cA_z \Theta + \Theta \cA_z^* + \cB_z J \cB_z^* & = 0,
    \qquad
    z \in \mU_N,\\
\label{xyCCRdot0}
    \sum_{z\in \mU_N}
    \cA_z^p
    (\Theta C^{\rT}+\cB_z J D^{\rT}) & = 0,
    \qquad
    p = 0,\ldots, n-1.
\end{align}
\end{thm}
\begin{proof}
In view of (\ref{XXxx}), the preservation of the CCRs (\ref{xCCR}) is equivalent to the preservation of CCRs (\ref{XCCR}) for the DFT of the dynamic variables.
Now, by employing the ideas of \cite[Proof of Theorem~2.1 on pp. 1798--1799]{JNP_2008},  the quantum Ito formula can be combined with the bilinearity of the commutator as
$
    \rd [\xi,\eta]
    =
    [\rd \xi, \eta] + [\xi, \rd \eta] + [\rd \xi,\rd \eta]
$ and applied to the left-hand side of (\ref{XCCR}), which  yields an ODE
\begin{align}
\nonumber
\rd [X_z,X_v^{\dagger} ]
     = &
     [\rd  X_z, X_v^{\dagger}]
     +
     [X_z, \rd X_v^{\dagger}]
     +
     [\rd X_z, \rd X_v^{\dagger}]          \\
\nonumber
     = &
     [
        \cA_z X_z \rd t + \cB_z \rd W_z,
        X_v^{\dagger}
    ] +
    [
        X_z,
        X_v^{\dagger}\cA_v^* \rd t + \rd W_v^{\dagger}\cB_v^*
    ]\\
\nonumber
      & +
    [
        \cA_z X_z \rd t + \cB_z \rd W_z,
        X_v^{\dagger}\cA_v^* \rd t + \rd W_v^{\dagger}\cB_v^*
    ]\\
\nonumber
     = &
    (\cA_z
    [
        X_z,
        X_v^{\dagger}
    ]
    +
    [
        X_z,
        X_v^{\dagger}
    ]
    \cA_v^*
    )\rd t+
    \cB_z
    [
        \rd W_z,
        \rd W_v^{\dagger} ]\cB_v^*\\
\label{dXCCR}
        = &
        2i         \delta_{zv}N
        (    \cA_z\Theta
    +
    \Theta
    \cA_v^*
    +
    \cB_z J\cB_v^*
    )\rd t.
\end{align}
Here, 
the commutativity between the adapted process $X_z$ and the forward increments $\rd W_z$  of the $z$-transformed quantum Wiener processes in (\ref{dX}) is used together with the quantum Ito product rules  \cite{P_1992} and (\ref{WCCR}). Now, for any two different roots $z \ne v$, the Kronecker delta $\delta_{zv}$ makes the right-hand side of (\ref{dXCCR}) vanish, and $[X_z, X_v^{\dagger}]$ is a conserved quantity in accordance with (\ref{XCCR}).  For equal roots $z=v \in \mU_N$, the ODE (\ref{dXCCR}) takes the form
$
    [X_z, X_z^{\dagger}]^{^\bullet}
        =
        2i
        N
        (    \cA_z\Theta
    +
    \Theta
    \cA_z^*
    +
    \cB_z J\cB_z^*
    )
$.
Therefore,
the preservation of the CCRs (\ref{XCCR}), or, equivalently, (\ref{xCCR}), holds  if and only if (\ref{Thetadot0}) is satisfied. Turning to (\ref{xyCCR}), we will now  consider the time evolution of $[X_z, Y_v^{\dagger}]$ for any fixed but otherwise arbitrary roots $z,v \in \mU_N$. Similarly to (\ref{dXCCR}), it follows that
\begin{align*}
\nonumber
\rd [X_z,Y_v^{\dagger} ]
     = &
     [\rd  X_z, Y_v^{\dagger}]
     +
     [X_z, \rd Y_v^{\dagger}]
     +
     [\rd X_z, \rd Y_v^{\dagger}]          \\
\nonumber
     = &
     [
        \cA_z X_z \rd t + \cB_z \rd W_z,
        Y_v^{\dagger}
    ] +
    [
        X_z,
        X_v^{\dagger}C^{\rT} \rd t + \rd W_v^{\dagger}D^{\rT}
    ]\\
\nonumber
      & +
    [
        \cA_z X_z \rd t + \cB_z \rd W_z,
        X_v^{\dagger}C^{\rT} \rd t + \rd W_v^{\dagger}D^{\rT}
    ]\\
\nonumber
     = &
    (\cA_z
    [
        X_z,
        Y_v^{\dagger}
    ]
    +
    [
        X_z,
        X_v^{\dagger}
    ]
    C^{\rT}
    )\rd t
    +
    \cB_z
    [
        \rd W_z,
        \rd W_v^{\dagger} ]D^{\rT}\\
\label{dXYCCR}
        = &
        (  \cA_z
    [
        X_z,
        Y_v^{\dagger}
    ]
    +
    2i\delta_{zv}N
    (
    \Theta
    C^{\rT}
    +
    \cB_z JD^{\rT})
    )\rd t,
\end{align*}
where use is also made of (\ref{WCCR}), (\ref{XCCR}).
Hence, the commutator $[X_z,Y_v^{\dagger}]$ satisfies a nonhomogeneous linear  ODE with constant coefficients and a constant forcing term:
$
        [
        X_z,
        Y_v^{\dagger}
    ]^{^\bullet}
        =
        \cA_z
    [
        X_z,
        Y_v^{\dagger}
    ]
    +
    2i\delta_{zv}N
    (
    \Theta
    C^{\rT}
    +
    \cB_z JD^{\rT})
$,
with zero initial condition (\ref{XYCCR0}). The solution of this initial value problem is
\begin{equation}
\label{XYCCRexp}
    [X_z(t),Y_v(t)^{\dagger}]
    =
    2i\delta_{zv}N
    \int_0^t
    \re^{\tau\cA_z}
    \rd \tau
    (\Theta C^{\rT}+ \cB_z JD^{\rT}),
\end{equation}
where, in the case $\det \cA_z \ne 0$, the matrix exponential  can be  integrated as
$
    \int_0^t
    \re^{\tau\cA_z}
    \rd \tau
    =
    \cA_z^{-1}
    (\re^{t\cA_z}-I_n)
$.
Since the right-hand side of (\ref{XYCCRexp}) vanishes for all $z\ne v$, the condition (\ref{xyXY0}) reduces to
\begin{equation}
\label{sumint}
    \sum_{z\in \mU_N}
    \int_0^t
    \re^{\tau\cA_z}
    \rd \tau
    (\Theta C^{\rT}+ \cB_z JD^{\rT})=0.
\end{equation}
In turn, the fulfillment of the latter equality at any time $t$ is equivalent to (\ref{xyCCRdot0}), which is obtained by repeatedly differentiating (\ref{sumint}) with respect to  $t$ and using the Cayley-Hamilton theorem.
\end{proof}

The following theorem reduces the CCR preservation conditions of Theorem~\ref{th:CCR} to a finite number of matrix algebraic equations for sufficiently large fragments of the chain.

\begin{thm}
\label{th:CCRexp}
The CCRs (\ref{xCCR}), (\ref{xyCCR}) for the dynamic and output variables of the systems $0, \ldots, N-1$  in the chain are preserved in time for any
\begin{equation}
\label{Nlarge}
    N\> \max(5,\, n+1)
\end{equation}
if and only if the state-space matrices  in (\ref{dx}), (\ref{dy}), (\ref{CDE}) satisfy
\begin{align}
\label{xCCRz0}
    A \Theta + \Theta A^{\rT} + BJB^{\rT} + E_+ D_+ JD_+^{\rT} E_+^{\rT} + E_- D_- JD_-^{\rT} E_-^{\rT} &= 0,\\
\label{xCCRz-1}
    E_+C_+\Theta + \Theta C_-^{\rT}E_-^{\rT} + BJD_-^{\rT}E_-^{\rT} + E_+D_+JB^{\rT} &= 0,\\
\label{xCCRz-2}
    E_+ D_+ J D_-^{\rT} E_-^{\rT} &= 0,\\
\label{xyCCRz00}
    \Theta C^{\rT} + BJD^{\rT} &= 0,\\
\label{xyCCRz0}
    (A_{p,1} E_+D_+ + A_{p,-1} E_- D_-)JD^{\rT} &= 0
\end{align}
for all $p=1,\ldots, n-1$. Here, the matrices $A_{p,s} \in \mR^{n\x n}$ are computed recursively as
\begin{equation}
\label{Anext}
    A_{p,s}
    =
        AA_{p-1,s} + E_+ C_+A_{p-1,s+1} + E_-C_- A_{p-1,s-1}
\end{equation}
for $s =0,\pm 1,\ldots, \pm p$, and $A_{p,s} := 0$ for $|s|>p$,
with initial condition $A_{0,0} = I_n$.
\end{thm}
\begin{proof}
By using the partitioning of the matrices $C$, $D$, $E$ in  (\ref{CDE}) and the matrix $K_z$ from (\ref{K}), it follows that the matrices $\cA_z$, $\cB_z$ in (\ref{cAB}) take the form
\begin{align}
\label{cA}
    \cA_z & = A + z^{-1} E_+ C_+ + z E_-C_-,\\
\label{cB}
    \cB_z & = B + z^{-1} E_+ D_+ + z E_-D_-.
\end{align}
Hence, in view of (\ref{star}), the left-hand side of (\ref{Thetadot0}) is a meromorphic function of $z$ whose Laurent series involves five powers $z^s$ with $s = 0, \pm 1, \pm 2$:
\begin{align}
\nonumber
      L_z
      := & \cA_z \Theta +  \Theta \cA_z^* + \cB_z J \cB_z^* \\
\nonumber
      = &
      (A + z^{-1} E_+ C_+ + z E_-C_-) \Theta\\
\nonumber
& + \Theta (A^{\rT} + z C_+^{\rT}E_+^{\rT} + z^{-1} C_-^{\rT} E_-^{\rT})\\
\nonumber
    & +
    (B + z^{-1} E_+ D_+ + z E_-D_-)\, J \\
\label{Laur1}
& \x (B^{\rT}+z D_+^{\rT}E_+^{\rT} + z^{-1} D_-^{\rT} E_-^{\rT})= \sum_{|s|\< 2} z^s M_s.
\end{align}
Here, the coefficients $M_s\in \mR^{n\x n}$ depend on the exponent $s$ in an antisymmetric fashion: $M_{-s} = - M_s^{\rT}$. Hence, in order for (\ref{Thetadot0}) to hold, it is sufficient that $M_0 = M_{-1} = M_{-2} = 0$. Direct calculation of these three matrices from (\ref{Laur1}) leads to (\ref{xCCRz0})--(\ref{xCCRz-2}). However, the relations (\ref{xCCRz0})--(\ref{xCCRz-2}) are necessary for (\ref{Thetadot0}) only if the size $N$ of the chain fragment is large enough. More precisely, the coefficients of the Laurent series in (\ref{Laur1})  can be uniquely reconstructed from the map $\mU_N \ni z \mapsto L_z$ as
\begin{equation}
\label{LM}
    M_s = \frac{1}{N}\sum_{z\in \mU_N} z^{-s} L_z,
    \qquad
    s = 0,\pm 1, \pm 2,
\end{equation}
provided $N> \max\{|q-s|:\, -2\< q,s\< 2\} =4$. This is a manifestation of the circular version \cite{S_1979} of the Whittaker-Kotel'nikov-Shannon sampling theorem \cite{J_1977} and follows from the identity
\begin{equation}
\label{zq}
    \frac{1}{N}
    \sum_{z\in \mU_N}
    z^q
    =
    \left\{
        \begin{array}{ccc}
            1 & {\rm if}\ q \equiv 0\, (\mod N)\\
            0 & {\rm otherwise}
        \end{array}
    \right..
\end{equation}
The latter holds for all integers $q$ and, in application to (\ref{Laur1}), yields
$$
    \frac{1}{N}
    \sum_{z\in \mU_N}
    z^{-s} L_z
    =
    \frac{1}{N}
    \sum_{q=-2}^2
    M_q
    \sum_{z\in \mU_N}
    z^{q-s}
    =
    \sum_{|q|\< 2:\ q\equiv s (\mod N)}
    M_q.
$$
Therefore, if $N\> 5$, then the fulfillment of $L_z = 0$ for all $z \in \mU_N$ in  (\ref{Thetadot0}) indeed implies that $M_s=0$ for all $s = 0,\pm 1, \pm 2$ in view of  (\ref{LM}), thus establishing the necessity of (\ref{xCCRz0})--(\ref{xCCRz-2}) for (\ref{Thetadot0}). We will now turn to the condition (\ref{xyCCRdot0}).
By applying (\ref{zq}) to a function $f$ of a complex variable with a finite Laurent series
$
    f(z) = \sum_{|s|\< p} c_sz^s
$
(whose pole orders at $z=0,\infty$ do not exceed $p$),
the average over the set of $N$th roots of unity reduces to a complex residue as
$$
    \frac{1}{N}
    \sum_{z\in \mU_N}
    f(z)
    =
    c_0
    =
    \frac{1}{2\pi i}
    \oint_{\mU}
    \frac{f(z)}{z}\rd z
    =
    \Res_{z = 0}\frac{f(z)}{z},
$$
provided $N > p$. By applying this ``averaged'' version of the circular sampling theorem to the functions
\begin{equation}
\label{Q}
    Q_p(z)
    :=
    \cA_z^p
    (\Theta C^{\rT}+\cB_z J D^{\rT})
\end{equation}
on the left-hand side of (\ref{xyCCRdot0}) with $p = 0,\ldots, n-1$, and using (\ref{cA}), (\ref{cB}),
it follows that for any $N > n$,
\begin{align}
\nonumber
    \frac{1}{N}
    \sum_{z\in \mU_N}
    Q_p(z)
    &=
    \Res_{z=0}\frac{Q_p(z)}{z}    \\
\nonumber
    &=\Res_{z=0}
        \sum_{|s|\< p}
        z^{s-1} A_{p,s}
        (\Theta C^{\rT}
        +
        (B + z^{-1} E_+ D_+ + z E_-D_-)JD^{\rT}
        )\\
\label{QQ}
    &=
    A_{p,0}(\Theta C^{\rT} + BJD^{\rT})
+
    (A_{p,1} E_+D_+
    +
    A_{p,-1} E_-D_-) JD^{\rT}.
\end{align}
Here, the matrices $A_{p,s}\in \mR^{n\x n}$ are defined as the coefficients of the Laurent series
\begin{equation}
\label{Aps}
    \cA_z^p = \sum_{|s|\< p} z^s A_{p,s}
\end{equation}
for the powers of the matrix $\cA_z$ from (\ref{cA})
and satisfy the recurrence relation (\ref{Anext}) initialized by $A_{0,0}=I_n$.
Therefore, (\ref{Q}), (\ref{QQ}) imply  that (\ref{xyCCRdot0}) is equivalent to
\begin{equation}
\label{AA}
    A_{p,0}(\Theta C^{\rT} + BJD^{\rT})
+
    (A_{p,1} E_+D_+
    +
    A_{p,-1} E_-D_-) JD^{\rT} = 0
\end{equation}
for all $p=0, \ldots, n-1$, provided $N > n$. Since $A_{0,0}=I_n$ and $A_{0,\pm 1} = 0$, then for $p=0$, the condition (\ref{AA}) takes the form of (\ref{xyCCRz00}). In turn,  (\ref{xyCCRz00}) reduces (\ref{AA}) to (\ref{xyCCRz0}) for $p=1, \ldots, n-1$.
It now remains to combine the inequalities $N\> 5$ and $N\> n+1$ into (\ref{Nlarge}) whose fulfillment ensures that (\ref{xCCRz0})--(\ref{xyCCRz0}) are equivalent to (\ref{Thetadot0}) and (\ref{xyCCRdot0}). \end{proof}

Application of Theorem~\ref{th:CCRexp} involves computation of the matrices $A_{p,s}$.
For example,     it follows directly from (\ref{cA}), (\ref{Aps}) that $A_{1,-1} = E_+C_+$, $A_{1,0} = A$, $
A_{1,1}= E_-C_-$, and the next iteration of (\ref{Anext}) yields
\begin{align*}
    A_{2,-2}& = (E_+C_+)^2,
    \qquad
    A_{2,-1}= A E_+C_+ + E_+C_+ A,\\
    A_{2,0}& = E_-C_- E_+C_+ + A^2 + E_+C_+E_-C_-,    \\
    A_{2,1}& = E_-C_- A + AE_-C_-,\qquad
    A_{2,2} = (E_-C_-)^2.
\end{align*}

\section{\bf Mean square performance with Toeplitz weights}\label{sec:quadro}

We will now consider a mean square functional defined  at time $t$ as the average of a quadratic form of dynamic variables associated with the fragment of the chain of size $N$:
\begin{equation}
\label{cE}
    \cE_N(t)
    :=
    \bE
    \left(
    \sum_{j,k=0}^{N-1}
    x_j(t)^{\rT} \sigma_{j-k} x_k(t)
    \right).
\end{equation}
Here, the quantum expectation
  $\bE \xi := \Tr (\rho \xi)$ of a quantum variable $\xi$ (such that $\rho \xi$ is a trace class operator) is taken over the density operator $\rho:= \rho_0\ox \upsilon$ which is the tensor product of an initial quantum state $\rho_0$ of the network and the vacuum state $\upsilon$ of the external fields \cite{P_1992}. Also, $\sigma_k$ is a given $\mR^{n\x n}$-valued sequence   which satisfies $\sigma_{-k} = \sigma_k^{\rT}$  for all integers $k$ and specifies a real symmetric block Toeplitz  weighting matrix $(\sigma_{j-k})_{0\< j,k<N}$. The block Toeplitz structure of the weighting matrix in (\ref{cE}) is in line with the translation invariance of the quantum network being considered. For what follows,
the weighting sequence is assumed to be absolutely summable, that is,
$    \sum_{k=-\infty}^{+\infty}\|\sigma_k\| < +\infty
$,
with $\|\cdot\|$ the Frobenius norm \cite{HJ_2007}, and hence, its (two-sided) Fourier transform is well-defined:
\begin{equation}
\label{Sigma}
  \Sigma_z := \sum_{k =-\infty}^{+\infty} z^{-k} \sigma_k,
  \qquad
  z \in \mU.
\end{equation}
The continuous matrix-valued map $\mU\ni z\mapsto \Sigma_z = \Sigma_z^*$
describes the spectral density of the weighting sequence.
The fulfillment of $\Sigma_z\succcurlyeq 0$ for all $z\in \mU$
is necessary and sufficient for 
$(\sigma_{j-k})_{0\< j,k<N}\succcurlyeq 0$ for all $N\> 1$; see, for example, \cite{GS_1958}. In this case, the self-adjoint operator, described by the quadratic form on the right-hand side of (\ref{cE}), is positive semi-definite, and  hence, $\cE_N \> 0$.
Mean square performance criteria are used in the CQLQG control/filtering problems \cite{NJP_2009,VP_2013a,VP_2013b}. The quantity $\cE_N$ in (\ref{cE}) will therefore be referred to as the current LQG cost of the network.
  Since we are concerned with linear dynamics and second-order moments of quantum variables, the initial network state $\rho_0$ is not required to be Gaussian \cite{KRP_2010}, and the use of the qualifier ``Gaussian'' is only a reference to  LQG formulations for classical systems \cite{AM_1989,KS_1972}.
The following theorem provides a stability criterion for the quantum network and computes the steady-state value of the LQG cost per site for large fragments of the lattice, which corresponds to the thermodynamic limit in equilibrium statistical mechanics \cite{R_1978}.

\begin{thm}
\label{th:LQGcost}
Suppose the matrix $\cA_z$ in (\ref{cAB}) is Hurwitz for all $z$ on the unit circle:
\begin{equation}
\label{stab}
    \max_{z\in \mU}\, \br\big(\re^{\cA_z}\big) <1.
\end{equation}
Then the LQG cost $\cE_N$ in (\ref{cE}) has an infinite-horizon limit per site for unboundedly increasing fragments of the chain:
\begin{equation}
\label{limlim}
    \lim_{N\to +\infty}
    \left(
        \frac{1}{N}
        \lim_{t\to +\infty}\cE_N(t)
    \right)
    =
    \frac{1}{2\pi i}
    \oint_{\mU}
    \Tr
    (
    \Sigma_z S_z)
    \frac{\rd z}{z}.
\end{equation}
Here, $\Sigma_z$ is the weighting spectral density in (\ref{Sigma}), and the matrix $S_z = S_z^*\succcurlyeq 0$
is the unique solution of 
the algebraic Lyapunov equation
\begin{equation}
\label{SLyap}
    \cA_z S_z + S_z\cA_z^* + \cB_z\Omega\cB_z^* = 0,
\end{equation}
where $\Omega$ is the quantum Ito matrix from (\ref{ww}), and $\cB_z$ is given by (\ref{cAB}).
\end{thm}

\begin{proof}
Similarly to (\ref{dXCCR}), for any given roots $z,v\in \mU_N$, the matrix $X_zX_v^{\dagger}$ satisfies a QSDE
\begin{align}
\nonumber
\rd (X_zX_v^{\dagger} )
     = &
     (\rd  X_z) X_v^{\dagger}
     +
     X_z\rd X_v^{\dagger}
     +
     (\rd X_z) \rd X_v^{\dagger}          \\
\nonumber
     = &
    (\cA_z X_zX_v^{\dagger} + X_zX_v^{\dagger} \cA_v^*)\rd t
    +
    \cB_z \rd W_z \rd W_v^{\dagger} \cB_v^*\\
\nonumber
    & +
    \cB_z \rd W_z X_v^{\dagger} + X_z \rd W_v^{\dagger} \cB_v^*\\
\nonumber
     = &
    (\cA_z X_zX_v^{\dagger} + X_zX_v^{\dagger} \cA_v^* + N \delta_{zv} \cB_z\Omega\cB_v^*)\rd t\\
\label{dXX}
    & +
    \cB_z \rd W_z X_v^{\dagger} + X_z \rd W_v^{\dagger} \cB_v^*.
\end{align}
Since the forward increments $\rd W_z$ of the quantum Wiener process   in the vacuum state are uncorrelated with the adapted processes $X_v$, then averaging of both sides of (\ref{dXX}) leads to a Lyapunov ODE
\begin{equation}
\label{Sdot}
    \dot{S}_{z,v} = \cA_z S_{z,v} + S_{z,v}\cA_v^* + N \delta_{z,v} \cB_z\Omega\cB_v^*,
\end{equation}
where
\begin{equation}
\label{S}
    S_{z,v}(t):= \bE(X_z(t)X_v(t)^{\dagger})
\end{equation}
is the matrix of second-order cross-moments of the vectors $X_z$ and $X_v$ at time $t\> 0$. The solution of (\ref{Sdot}) is described by
$
    S_{z,v}(t) = \re^{t\cA_z}S_{z,v}(0)\re^{t\cA_v^*}
    +
    N \delta_{zv} \int_0^t \re^{\tau\cA_z}\cB_z\Omega\cB_v^*    \re^{\tau\cA_v^*}\rd \tau
$,
and, under the stability assumption (\ref{stab}), has the following limit
\begin{equation}
\label{Slim}
    S_{z,v}(\infty)
    :=
    \lim_{t\to +\infty}
    S_{z,v}(t) = N \delta_{zv} S_z,
\end{equation}
where
\begin{equation}
\label{Sz}
    S_z := \int_{0}^{+\infty} \re^{t\cA_z}\cB_z\Omega\cB_z^*    \re^{t\cA_z^*}\rd t
\end{equation}
is a complex positive semi-definite Hermitian matrix which is
the unique solution of (\ref{SLyap}). Note that $z\mapsto S_z$ is a continuous function on the unit circle.
Now, by combining the inverse DFT of $X_z$ in (\ref{invFourier})  with (\ref{S}), it follows that
\begin{align}
\nonumber
    \bE \big(x_j^{\rT}\sigma_{j-k} x_k\big)
    & =
    \Tr\big(
    \sigma_{j-k}^{\rT} \bE\big(x_jx_k^{\rT}\big)\big)\\
\label{xsigx}
& =
    \frac{1}{N^2}
    \sum_{z,v \in \mU_N}
    z^jv^{-k}
    \Tr
    \big(
    \sigma_{k-j} S_{z,v}\big),
\end{align}
where use is also made of the property $\sigma_k^{\rT} = \sigma_{-k}$ of the weighting sequence.
In view of (\ref{Slim}), the substitution of (\ref{xsigx}) into (\ref{cE}) leads to the following infinite-horizon limit for the LQG cost:
\begin{align}
\nonumber
    \cE_N(\infty)
    & :=
    \lim_{t\to +\infty}
    \cE_N(t)
 =
    \frac{1}{N^2}
    \sum_{j,k=0}^{N-1}
    \sum_{z,v \in \mU_N}
    z^jv^{-k}
    \Tr
    \big(
    \sigma_{k-j} S_{z,v}(\infty)\big)\\
\label{cElim}
    & =
    \frac{1}{N}
    \sum_{j,k=0}^{N-1}
    \sum_{z \in \mU_N}
    z^{j-k}
    \Tr
    \big(
    \sigma_{k-j} S_z\big)
    =
    \sum_{z \in \mU_N}
    \Tr
    \big(
    \wh{\Sigma}_N(z) S_z\big),
\end{align}
where
$    \wh{\Sigma}_N(z) := \sum_{\ell = 1-N}^{N-1}\left(1 - \frac{|\ell|}{N}\right)z^{-\ell} \sigma_{\ell}
$.
Due to the absolute summability of the weighting sequence, the function $\wh{\Sigma}_N(z)$ converges to $\Sigma_z$ in (\ref{Sigma}) uniformly over the unit circle as $N\to +\infty$:
\begin{equation}
\label{SS}
    \max_{z\in \mU}\|\wh{\Sigma}_N(z)-\Sigma_z\|
    \<
    2\sum_{\ell=1}^{+\infty}
    \|\sigma_{\ell}\|
    \min\Big(1, \frac{\ell}{N}\Big)
    \to 0,
\end{equation}
where the convergence to $0$ is obtained via the Lebesgue dominated convergence theorem. In view of the continuous dependence (and hence, boundedness) of $S_z$ on $z \in \mU$ in (\ref{Sz}) and the uniform convergence (\ref{SS}), it follows from (\ref{cElim}) that
\begin{equation}
\label{SSS}
    \left|
    \frac{\cE_N(\infty)}{N} - \frac{1}{N} \sum_{z\in \mU_N}    \Tr
    \big(
    \Sigma_z S_z\big)
    \right|
    \<
    \max_{z\in \mU}\|S_z\|        \,
    \max_{v\in \mU}\|\wh{\Sigma}_N(v)-\Sigma_v\|
    \to 0,
    \qquad
    N \to +\infty,
\end{equation}
where use is also made of the Cauchy-Bunyakovsky-Schwarz inequality for the Frobenius inner product of matrices. It now remains to note that $\frac{1}{N} \sum_{z\in \mU_N}    \Tr
    (\Sigma_z S_z)$ is  a Riemann sum which converges to the integral
    $$
        \frac{1}{2\pi}
        \int_{-\pi}^{\pi}
    \Tr
    \big(
    \Sigma_{\re^{i\varphi}} S_{\re^{i\varphi}}\big)
    \rd \varphi
    =
        \frac{1}{2\pi i}
    \oint_{\mU}
    \Tr
    (
    \Sigma_z S_z)
    \frac{\rd z}{z}.
    $$
    In combination with (\ref{SSS}), this implies the convergence of $\cE_N(\infty)/N$ to the same integral as $N\to +\infty$, thus establishing (\ref{limlim}).
\end{proof}

The proof of Theorem~\ref{th:LQGcost} shows that the continuous matrix-valued map  $\mU \ni z \mapsto S_z=S_z^*\succcurlyeq 0$, computed through (\ref{SLyap}), is a spatial spectral density of the dynamic variables of the quantum network in the thermodynamic limit:
$$
    \lim_{N \to +\infty}\,
    \lim_{t\to +\infty}
    \bE\big(x_j(t)x_k(t)^{\rT}\big)
    =
    \frac{1}{2\pi i}
    \oint_{\mU}
    z^{j-k-1}S_z \rd z.
$$
In this asymptotic sense, the spectral density encodes
the covariance structure of the dynamic variables, which is closely related to the Grenander-Szeg\"{o} limit
theorem for Toeplitz forms
\cite{GS_1958}.
In view of condition (\ref{Thetadot0}) of Theorem~\ref{th:CCR},
the  spectral density involves
the common CCR matrix of component systems in the network as
$
    \Im S_z = \Theta
$. Therefore, the fact, that
$S_z = \Re S_z + i \Theta \succcurlyeq 0$ for all $z \in \mU$, is
a generalized Heisenberg uncertainty principle \cite{H_2001} in the spatial frequency domain.

\section{\bf Two-dimensional lattice of linear quantum stochastic systems}\label{sec:2dlattice}

The results of the previous sections can be extended to networks of quantum systems with nearest neighbour interaction on lattices of high  dimension.
The constituent blocks of such a network are labelled by multiindices, and the parameter of the spatial DFT becomes multivariate. 
For better visualizability, we will consider a two-dimensional lattice of linear quantum stochastic systems  which is shown in Figs.~\ref{fig:net}, \ref{fig:netblock} and is governed by a set of QSDEs
\begin{align}
\label{dxnet}
    \rd x_{jk}
    &=
    A x_{jk} \rd t + B \rd w_{jk} + E \rd u_{jk},\\
\label{dynet}
    \rd y_{jk}
    &=
    C x_{jk} \rd t + D \rd w_{jk}.
\end{align}
\begin{figure}[htbp]
\centering
\includegraphics[width=12cm]{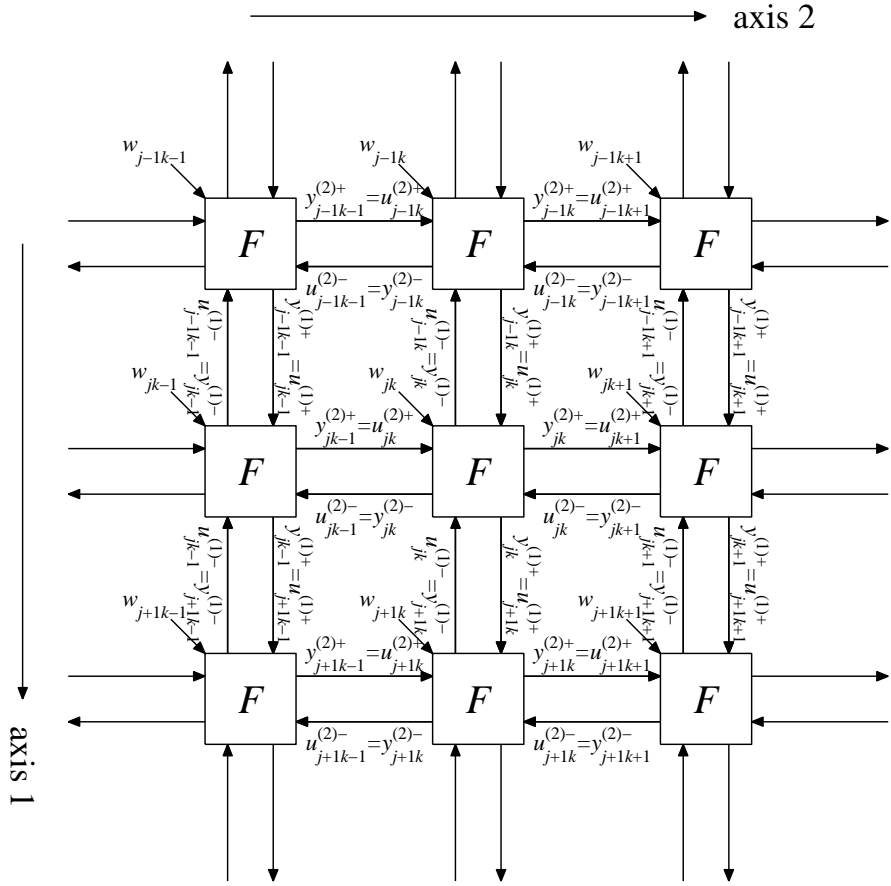}\vskip5mm
\caption{A fragment of a two-dimensional lattice of linear quantum stochastic systems with nearest neighbour interaction.}
    \label{fig:net}
\end{figure}
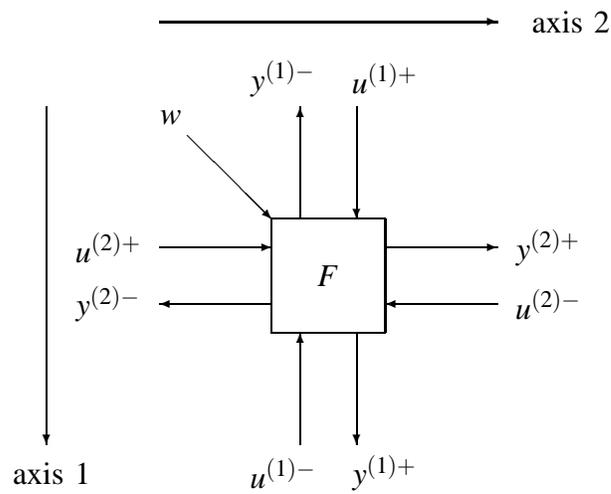
\begin{figure}[htbp]
\centering
\unitlength=0.75mm
\linethickness{0.4pt}
\begin{picture}(70.00,70.00)
    \put(20, 25){
        \framebox(20,20)[cc]{{$F$}}
       \put(-60,40){\vector(0,-1){60}}
       \put(-40,55){\vector(1,0){60}}
        \put(-60,-25){\makebox(0,0)[cc]{ axis 1}}
        \put(32,55){\makebox(0,0)[cc]{ axis 2}}

        \put(-5,40){\vector(0,-1){20}}
        \put(-5,0){\vector(0,-1){20}}

        \put(-15,20){\vector(0,1){20}}
        \put(-15,-20){\vector(0,1){20}}

        \put(-20,5){\vector(-1,0){20}}
        \put(20,5){\vector(-1,0){20}}
        \put(-40,15){\vector(1,0){20}}
        \put(0,15){\vector(1,0){20}}

        \put(-38,38){\makebox(0,0)[cc]{{$w$}}}
        \put(-35,35){\vector(1,-1){15}}



        \put(-49,15){\makebox(0,0)[cc]{{$u^{(2)+}$}}}
        \put(-49,5){\makebox(0,0)[cc]{{$y^{(2)-}$}}}

        \put(29,15){\makebox(0,0)[cc]{{$y^{(2)+}$}}}
        \put(29,5){\makebox(0,0)[cc]{{$u^{(2)-}$}}}

        \put(0,-25){\makebox(0,0)[cc]{{$y^{(1)+}$}}}
        \put(-18,-25){\makebox(0,0)[cc]{{$u^{(1)-}$}}}

        \put(0,45){\makebox(0,0)[cc]{{$u^{(1)+}$}}}
        \put(-18,45){\makebox(0,0)[cc]{{$y^{(1)-}$}}}

    }
\end{picture}\vskip5mm
\caption{A common building block of the two-dimensional lattice of systems.
}
\label{fig:netblock}
\end{figure}
\noindent Here, the pairs of indices $0\< j,k < N$ label the constituent systems with $n$-dimensional vectors $x_{jk}$ of dynamic variables  in a square fragment of the lattice of size  $N\x N$. Also,
\begin{equation}
\label{uynet}
    u_{jk}
    :=
    {\begin{bmatrix}
        u_{jk}^{(1)}\\
        u_{jk}^{(2)}
    \end{bmatrix}},
    \qquad
    u_{jk}^{(\alpha)}
    :=
    {\begin{bmatrix}
        u_{jk}^{(\alpha)+}\\
        u_{jk}^{(\alpha)-}
    \end{bmatrix}},
    \qquad
    y_{jk}
    :=
    {\begin{bmatrix}
        y_{jk}^{(1)}\\
        y_{jk}^{(2)}
    \end{bmatrix}},
    \qquad
    y_{jk}^{(\alpha)} :=
    {\begin{bmatrix}
        y_{jk}^{(\alpha)+}\\
        y_{jk}^{(\alpha)-}
    \end{bmatrix}},
\end{equation}
where $u_{jk}^{(\alpha)}$ and $y_{jk}^{(\alpha)}$ denote the vectors of
input and output quantum processes  (of common dimension $
    m^{(\alpha)}:= m_+^{(\alpha)}+m_-^{(\alpha)}
$) along the quantum channels parallel to the $\alpha$th reference axis. Accordingly, $A \in \mR^{n\x n}$, $B \in \mR^{n \x m_0}$ are given real matrices as before,  whilst the partitioning (\ref{CDE}) is replaced with
\begin{align}
\label{CDE12}
    C & :=
    {\begin{bmatrix}
        C^{(1)} \\
        C^{(2)}
    \end{bmatrix}},
    \qquad\quad\ \,
    D
    :=
    {\begin{bmatrix}
        D^{(1)} \\
        D^{(2)}
    \end{bmatrix}},
    \qquad\quad\
    E :=
    {\begin{bmatrix}
        E^{(1)} & E^{(2)}
    \end{bmatrix}},\\
\label{CDE+-}
    C^{(\alpha)}
    & :=
    {\begin{bmatrix}
        C_+^{(\alpha)} \\
        C_-^{(\alpha)}
    \end{bmatrix}},
    \qquad
    D^{(\alpha)}
    :=
    {\begin{bmatrix}
        D_{+}^{(\alpha)} \\
        D_{-}
    \end{bmatrix}},
    \qquad
    E^{(\alpha)} :=
    {\begin{bmatrix}
        E_+^{(\alpha)} & E_-^{(\alpha)}
    \end{bmatrix}},
\end{align}
where $C_{\pm}^{(\alpha)} \in \mR^{m_{\pm}^{(\alpha)} \x n}$, $D_{\pm}^{(\alpha)} \in \mR^{m_{\pm}^{(\alpha)} \x m_0}$, $E_{\pm}^{(\alpha)}\in \mR^{n\x m_{\pm}^{(\alpha)}}$  are associated with the $\alpha$th reference axis, $\alpha = 1,2$.
In addition to the QSDEs (\ref{dxnet}), (\ref{dynet}), the inputs and outputs of adjacent systems in the lattice are related by
\begin{equation}
\label{connectnet12}
    y_{j-1,k}^{(1)+}
    = u_{jk}^{(1)+},
    \qquad
    u_{jk}^{(1)-} = y_{j+1,k}^{(1)-},
    \qquad
    y_{j,k-1}^{(2)+}
    = u_{jk}^{(2)+},
    \qquad
    u_{jk}^{(2)-} = y_{j,k+1}^{(2)-},
\end{equation}
which corresponds to (\ref{connect}) along the reference axes;
see Fig.~\ref{fig:net}. The equations (\ref{dxnet}), (\ref{dynet}), (\ref{connectnet12})   for the $N\x N$-fragment of the lattice are equipped with a bivariate version of the PBCs (\ref{PBC}):
\begin{equation}
\label{PBC12}
     u_{0k}^{(1)+} = y_{N-1,k}^{(1)+},
    \qquad
    u_{N-1,k}^{(1)-} = y_{0k}^{(1)-},
    \qquad
    u_{j0}^{(2)+} = y_{j,N-1}^{(2)+},
    \qquad
    u_{j,N-1}^{(2)-} = y_{j0}^{(2)-}
\end{equation}
which are to be satisfied for all $ 0\< j,k< N$.
As in the one-dimensional case, the dynamics of the quantum network are encoded in a linear quantum stochastic system $F_z$ described in Sec. \ref{sec:ztrans} and governed by the QSDEs (\ref{dY}), (\ref{dX1}), except that the DFTs $X_z$, $W_z$, $U_z$, $Y_z$ in (\ref{XW})--(\ref{Y}) are now parameterized by points $z:= (z_1,z_2)\in \mU_N^2$ of a torus $\mU^2$. For example,
$$
    X_z(t)
    :=
    \sum_{j,k=0}^{N-1} z_1^{-j}z_2^{-k}x_{jk}(t).
$$
In view of (\ref{uynet}), (\ref{connectnet12}), (\ref{PBC12}), the matrix $K_z$ for the static unitary relation between $U_z$ and $Y_z$ in  (\ref{K}) is modified appropriately:
$$    K_z
    :=
    {\begin{bmatrix}
        K_{z_1}^{(1)} & 0\\
        0 & K_{z_2}^{(2)}
    \end{bmatrix}},
    \qquad
    K_v^{(\alpha)}
    :=
    {\begin{bmatrix}
        v^{-1}I_{m_+^{(\alpha)}} & 0\\
        0 & vI_{m_-^{(\alpha)}}
    \end{bmatrix}}.
$$
This corresponds to the block diagram in Fig. \ref{fig:netblockz}
\begin{figure}[htbp]
\centering
\unitlength=0.75mm
\linethickness{0.4pt}
\begin{picture}(70.00,95.00)
    \put(20, 35){
        \framebox(20,20)[cc]{{$F$}}

        \put(-5,45){\vector(0,-1){25}}
        \put(-5,0){\vector(0,-1){25}}

        \put(-15,20){\vector(0,1){25}}
        \put(-15,-25){\vector(0,1){25}}

        \put(-20,5){\line(-1,0){14}}
        \put(-36,5){\vector(-1,0){9}}
        \put(25,5){\line(-1,0){9}}
        \put(14,5){\vector(-1,0){14}}
        \put(-45,15){\line(1,0){9}}
        \put(-34,15){\vector(1,0){14}}
        \put(0,15){\line(1,0){14}}
        \put(16,15){\vector(1,0){9}}

        \put(-29,29){\makebox(0,0)[cc]{{$W_z$}}}
        \put(-26,26){\vector(1,-1){6}}



        \put(-55,15){\makebox(0,0)[cc]{{$U_z^{(2)+}$}}}
        \put(-55,5){\makebox(0,0)[cc]{{$Y_z^{(2)-}$}}}

        \put(20,30){\framebox(10,10)[cc]{{$z_2^{-1}$}}}
        \put(25,15){\vector(0,1){15}}
        \put(20,35){\line(-1,0){4}}
        \put(14,35){\line(-1,0){18}}
        \put(-6,35){\line(-1,0){8}}
        \put(-16,35){\line(-1,0){18}}
        \put(-36,35){\vector(-1,0){9}}
        \put(-45,35){\vector(0,-1){20}}

        \put(-45,5){\vector(0,-1){15}}
        \put(-50,-20){\framebox(10,10)[cc]{{$z_2$}}}
        \put(-40,-15){\line(1,0){4}}
        \put(-34,-15){\line(1,0){18}}
        \put(-14,-15){\line(1,0){8}}
        \put(-14,-15){\line(1,0){8}}
        \put(-4,-15){\line(1,0){18}}
        \put(16,-15){\vector(1,0){9}}
        \put(25,-15){\vector(0,1){20}}

        \put(35,15){\makebox(0,0)[cc]{{$Y_z^{(2)+}$}}}
        \put(35,5){\makebox(0,0)[cc]{{$U_z^{(2)-}$}}}

        \put(-5,-25){\vector(1,0){15}}
        \put(10,-30){\framebox(10,10)[cc]{{$z_1^{-1}$}}}
        \put(15,-20){\vector(0,1){65}}
        \put(15,45){\vector(-1,0){20}}

        \put(-15,45){\vector(-1,0){15}}
        \put(-40,40){\framebox(10,10)[cc]{{$z_1$}}}
        \put(-35,40){\vector(0,-1){65}}
        \put(-35,-25){\vector(1,0){20}}

        \put(-1,-32){\makebox(0,0)[cc]{{$Y_z^{(1)+}$}}}
        \put(-17,-32){\makebox(0,0)[cc]{{$U_z^{(1)-}$}}}

        \put(-1,52){\makebox(0,0)[cc]{{$U_z^{(1)+}$}}}
        \put(-17,52){\makebox(0,0)[cc]{{$Y_z^{(1)-}$}}}

    }
\end{picture}\vskip5mm
\caption{A block diagram of the $z$-transformed QSDEs for the two-dimensional lattice. }
\label{fig:netblockz}
\end{figure}
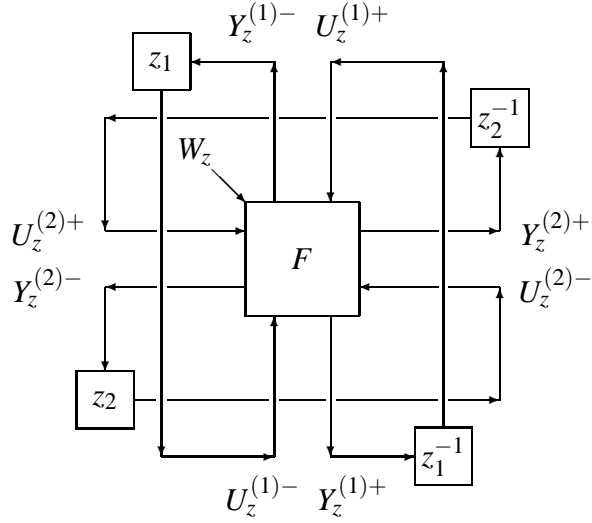
and, in combination with (\ref{CDE12}), (\ref{CDE+-}), yields the following state-space matrices  in the spatial frequency domain:
\begin{align}
\label{cAnet}
    \cA_z & = A + \sum_{\alpha=1}^2 \left(z_{\alpha}^{-1} E_+^{(\alpha)} C_+^{(\alpha)} + z_{\alpha} E_-^{(\alpha)}C_-^{(\alpha)}\right),\\
\label{cBnet}
    \cB_z & = B + \sum_{\alpha=1}^2 \left(z_{\alpha}^{-1} E_+^{(\alpha)} D_+^{(\alpha)} + z_{\alpha} E_-^{(\alpha)}D_-^{(\alpha)}\right).
\end{align}
In the bivariate case, the CCRs (\ref{xCCR}), (\ref{xyCCR}) are replaced with
$$
    [x_{jk}, x_{\ell s}^{\rT}] = 2i\delta_{jk}\delta_{\ell s}\Theta,
    \qquad
    [x_{jk}, y_{jk}^{\rT}]    = 0,
    \qquad
    0\< j,k,\ell,s <N.
$$
In this case, Theorem~\ref{th:CCR} remains valid, except that the conditions (\ref{Thetadot0}), (\ref{xyCCRdot0}) are to be understood as their torus counterparts (with $\mU_N^2$ instead of  $\mU_N$) and applied to the modified matrices $\cA_z$, $\cB_z$ in (\ref{cAnet}), (\ref{cBnet}). Upon these modifications, Theorem~\ref{th:LQGcost} also extends to the bivariate case, with the spectral densities being defined on the torus $\mU^2$. However, the above  modifications
affect a particular form of Theorem~\ref{th:CCRexp} whose multivariate analogue will be discussed elsewhere.

\section*{\bf Acknowledgement}

IGV thanks Dr Matthew J. Woolley for useful discussions on optical lattices.

\end{document}